\documentclass[amstex]{article}
\usepackage{amsmath}
\usepackage{amssymb}
\usepackage{amsfonts}
\begin{document}
\newtheorem{thm}{Theorem}[section]
\newtheorem{lem}[thm]{Lemma}
\newtheorem{prop}[thm]{Proposition}
\newtheorem{cor}[thm]{Corollary}
\newtheorem{assum}[thm]{Assumption}
\newtheorem{rem}[thm]{Remark}
\newtheorem{defn}[thm]{Definition}
\newtheorem{clm}[thm]{Claim}
\newcommand{\lv}{\left \vert}
\newcommand{\rv}{\right \vert}
\newcommand{\lV}{\left \Vert}
\newcommand{\rV}{\right \Vert}
\newcommand{\la}{\left \langle}
\newcommand{\ra}{\right \rangle}
\newcommand{\ltm}{\left \{}
\newcommand{\rtm}{\right \}}
\newcommand{\lcm}{\left [}
\newcommand{\rcm}{\right ]}
\newcommand{\ket}[1]{\lv #1 \ra}
\newcommand{\bra}[1]{\la #1 \rv}
\newcommand{\lmk}{\left (}
\newcommand{\rmk}{\right )}
\newcommand{\al}{{\mathcal A}}
\newcommand{\md}{M_d({\mathbb C})}
\newcommand{\ali}[1]{{\mathfrak A}_{[ #1 ,\infty)}}
\newcommand{\alm}[1]{{\mathfrak A}_{(-\infty, #1 ]}}
\newcommand{\nn}[1]{\lV #1 \rV}
\newcommand{\br}{{\mathbb R}}
\newcommand{\dm}{{\rm dom}\mu}
\newcommand{\Ad}{\mathop{\mathrm{Ad}}\nolimits}
\newcommand{\Proj}{\mathop{\mathrm{Proj}}\nolimits}
\newcommand{\RRe}{\mathop{\mathrm{Re}}\nolimits}
\newcommand{\RIm}{\mathop{\mathrm{Im}}\nolimits}
\newcommand{\Tr}{\mathop{\mathrm{Tr}}\nolimits}
\newcommand{\Mat}{\mathop{\mathrm{Mat}}\nolimits}
\newcommand{\zin}{\mathbb{Z}}
\newcommand{\rr}{\mathbb{R}}
\newcommand{\cc}{\mathbb{C}}
\newcommand{\nan}{\mathbb{N}}
\newcommand{\wks}{\mathop{\mathrm{wk^*-}}\nolimits}
\def\qed{{\unskip\nobreak\hfil\penalty50
\hskip2em\hbox{}\nobreak\hfil$\square$
\parfillskip=0pt \finalhyphendemerits=0\par}\medskip}
\def\proof{\trivlist \item[\hskip \labelsep{\bf Proof.\ }]}
\def\endproof{\null\hfill\qed\endtrivlist\noindent}
\def\proofof[#1]{\trivlist \item[\hskip \labelsep{\bf Proof of #1.\ }]}
\def\endproofof{\null\hfill\qed\endtrivlist\noindent}

\title{The Shannon-McMillan Theorem for AF $C^*$-systems}
\author{Yoshiko Ogata\thanks
{Graduate School of Mathematics, University of Tokyo, Japan}}
\maketitle{}
\begin{abstract}We give a new proof of quantum
Shannon-McMillan theorem, extending it 
to AF $C^*$-systems.
Our proof is based on the variational principle, 
instead of the classical
Shannon-McMillan theorem.
\end{abstract}
\section{Introduction}
The classical Shannon-McMillan Theorem \cite{sha}
states that an ergodic system has {\it typical sets}
satisfying the
 asymptotic equipartition property.
This theorem demonstrates 
the significance of the entropy which
gives the {\it size} of
the typical sets. 
There has recently been great progress
in the quantum version of the
Shannon-McMillan Theorem 
\cite{hp2}, \cite{ns}, \cite{qsm}, \cite{bs}.
In particular, Bjelakovi\'{c} et al. \cite{qsm} proved 
 Shannon-McMillan theorem for ergodic 
quantum spin systems.
The analyses in \cite{hp2}, \cite{ns}, 
\cite{qsm}, \cite{bs} are based on the classical 
Shannon McMillan theorem. The quantum mean entropy
of a translation invariant state 
can be approximated by classical mean entropies of
its restriction to some abelian subalgebras.
This fact enables us to reduce the problem 
to the classical one. 

In this paper, we present an alternative
proof of quantum Shannon-McMillan theorem.
Our proof is based on the variational principle, 
which is a well-known
thermodynamic property of quantum spin systems.
Roughly speaking, the variational principle 
enables us to estimate
{\it rank} of support projections of ergodic states,
in terms of the mean entropy.
By virtue of this estimate, we are able to
prove quantum Shannon-McMillan theorem directly, without
relying on the classical version of it.
Using this argument, we extend the quantum Shannon-McMillan theorem to
 AF $C^*$-systems.
Our proof applies to any dynamical system which admits
 thermodynamical formalism. In particular, we can apply it
to quantum spin systems
on $\zin^\nu$-lattice with $\nu\ge 2$.
However, for the sake simplicity, in this paper,
we present the result
for $\zin$-action.

We consider quadruples $(A,\{A_{[n,m]}\}_{n\le m},\tau, \gamma)$,
$n,m\in{\mathbb Z}$, where $A$ is a unital $C^*$-algebra,
$A_{[n,m]}$ a finite-dimensional $C^*$-subalgebra of
$A$ with $1_A\in A_{[n,m]}$, $\tau$ a faithful trace state of $A$, and
$\gamma$ a $\tau$-preserving automorphism of $A$.
For any subset $\Lambda$ of $\mathbb Z$,
we denote by $A_{\Lambda}$ the $C^*$-subalgebra of
$A$ generated by $A_{[k,n]}, [k,n]\subset \Lambda$, 
and write $A_n$ instead of $A_{[0,n-1]}$.
We denote the set of all self-adjoint elements of $A$ by
$A_{\rm sa}$ and the set of all $\gamma$-invariant states
on $A$ by ${\cal S}_{\gamma}$.
For a family of intervals $\{I_{\alpha}\}_{\alpha}$ in $\zin$, 
we will write $\bigvee_{\alpha}A_{I_{\alpha}}$ for
the $C^*$-subalgebra generated by 
$\{A_{I_{\alpha}}\}_{\alpha}$.
If $B$ is a finite-dimensional $C^*$-subalgebra of $A$, its
canonical trace is denoted by $\Tr_B$. The restriction of a state
$\varphi$ on $A$ to $B$ is written $\varphi\vert_B$.
For a state $\psi$ on $B$, we denote its density matrix 
by $D_{\psi}$, and von Neumann entropy by $S(\psi)$.
For positive linear functionals $\psi_1,\psi_2$ on $B$,
the relative entropy is denoted by $S(\psi_1,\psi_2)$.
Throughout the paper we suppose that the following 
conditions are satisfied.
\begin{assum}\label{nine}
\begin{description}
\item[(i)]There exists $n_0>0$ such that 
$A_{(-\infty,0]}$ and $A_{[n_0,\infty)}$ commute;
\item[(ii)]$\tau(ab)=\tau(a)\tau(b)$ for 
$a\in A_{(-\infty,0]}$, $b\in A_{[n_0,\infty)}$;
\item[(iii)]$A_{[n',m']}\subset A_{[n'',m'']}$, for
$n''\le n'\le m'\le m''$;
\item[(iv)]$\bigcup_{n}A_{[-n,n]}$
is dense in $A$;
\item[(v)]$\gamma(A_{[n,m]})=A_{[n+1,m+1]}$
for all $n\le m$, $n,m\in{\mathbb Z}$;
\item[(vi)]for the density matrix $D_{\tau{\vert}_{A_n}}$ the limit
$\lambda_{\tau}:=\lim_{n\to\infty}-\frac 1 n \log D_{\tau{\vert}_{A_n}}$ exists in norm, 
and is a scalar.
\end{description}
\end{assum}
This class of systems were studied in \cite{hp},
\cite{gn}. Clearly, quantum spin chains belong to the class.
See \cite{hp},
\cite{gn} for the other examples.

For such a system, the mean entropy exists.  
\begin{prop}\label{me}
For any $\gamma$-invariant
state $\omega$, the limit 
\[
s(\omega):=\lim_{n}\frac1n S(\omega\vert_{A_n})
\]
exists in $[0, \lambda_{\tau}]$.
The function ${\cal S}_{\gamma}\ni \omega 
\mapsto s(\omega)\in [0, \lambda_{\tau}]$ is
affine and weakly$^*$ upper semicontinuous.
\end{prop}
This system also satisfies the variational principle.
\begin{prop}\label{vp}
Let $a$ be a self-adjoint element in $A$. Then the limit
\[
P_{\gamma}^{\tau}(a):=\lim_{n\to\infty} \frac 1n \log 
\tau(e^{-\sum_{i=0}^{n-1}\gamma_i(a)})\] exists and
\[
P_{\gamma}^{\tau}(a)
=\sup_{\omega\in {\cal S}_{\gamma}}\left\{
s(\omega)-\omega(a)\right\}-\lambda_{\tau}.
\]
\end{prop}
Proposition \ref{me} and \ref{vp} are well-known for quantum spin systems.
(See \cite{BR96}). For AF $C^*$-systems, $n_0=1$ case  was shown in 
\cite{hp}.
We give a proof for the general case in Section \ref{mevp}.

In this paper, we give a proof of 
quantum Shannon-McMillan Theorem in AF $C^*$-systems
via the variational principle, without relying on the classical
Shannon-McMillan Theorem.
\begin{thm}\label{main}
Let $(A,\{A_{[n,m]}\}_{n\le m},\tau, \gamma)$ be an AF
$C^*$-system satisfying Assumption \ref{nine}.
Then for any $\gamma$-ergodic state $\omega$ on $A$ and $\delta>0$, 
there exists a 
sequence of projections $\{p_n\}_{n\in {\mathbb N}}$ 
in $A$ with $p_n\in A_n$,
such that
\begin{description}
\item[(i)]$\lim_{n\to\infty}\omega(p_n)=1$;
\item[(ii)] for all minimal projections $e\in A_n$ with
$e\le p_n$, 
\[
e^{-n\lmk s(\omega)+\delta \rmk}\le
\omega(e)\le e^{-n\lmk s(\omega)-\delta \rmk};
\]
\item[(iii)]for $n$ large enough,
\[
e^{n(s(\omega)-\delta)}\le
Tr_{A_n}p_n\le e^{n(s(\omega)+\delta)}.
\]
\end{description}
\end{thm}
\section{Shannon-McMillan Theorem for AF $C^*$-systems}
For a self-adjoint element $a$ in a finite dimensional
$C^*$-algebra and a Borel set $I$ of $\mathbb R$,
$\Proj[a\in I]$ denotes the spectral projection of $a$ 
associated with
$I$.\\
\begin{lem}\label{lb}
For any $\gamma$-ergodic state $\omega$ and any
\begin{align}
t<-s(\omega)+\lambda_{\tau},\label{tu}
\end{align}
we have
\[
\lim_{n\to\infty}\omega\lmk\Proj\left
[\frac 1n \lmk\log D_{\omega\vert_{A_n}}-\log D_{\tau\vert_{A_n}}\rmk 
\le t\right] \rmk=0.
\]
\end{lem}
\begin{proof}
By Proposition \ref{vp} 
and concavity and upper semi-continuity 
of the mean entropy (Proposition \ref{me}), we have
\begin{align*}
s(\omega)=\inf_{a\in A_{sa}}\left\{
P_{\gamma}^{\tau}(a)+\omega(a)+\lambda_{\tau}\right\}.
\end{align*}
(See Theorem 3.12 of \cite{rueb}.)
Therefore, (\ref{tu}) means there exists a self-adjoint element $a\in A$
such that
\begin{align*}
t<-\lmk P_{\gamma}^{\tau}(a)+\omega(a)\rmk.
\end{align*}
Since the right hand side is continuous with respect 
to the norm of $a$ and invariant under the translation 
$a\to \gamma_i(a)$, we may assume $a$ is in $A_l$ for some
$l\in {\mathbb N}$.
For each $l\le n\in {\mathbb N}$, we set
$t_{n}(a):=\sum_{i=0}^{n-l}\gamma_i(a)/n\in A_n$.

Take $0<\delta<-t-\lmk P_{\gamma}^{\tau}(a)+\omega(a)\rmk$, 
and define
\[
F^n:=\Proj\left[
t_{n}(a)\in \lmk \omega(a)-\delta,\omega(a)+\delta\rmk
\right]\in A_n
\]
for each $n\in{\mathbb N}$ with $l\le n$.

By the ergodicity of $\omega$, we have
\begin{align}\label{oe}
\lim_{n\to\infty}\omega\lmk F^n\rmk=1.
\end{align}
(See Theorem 4.3.17 of \cite{BR86}.)
We claim
\begin{align}\label{rk}
\limsup_{n\to\infty}\frac 1n\log\tau\lmk
F^n\rmk\le P_{\gamma}^{\tau}(a)+\omega(a)+\delta.
\end{align}
This is trivial if $F^n=0$ eventually.
By taking a subsequence if necessary,
 we may assume $F^n\neq 0$ for all $n$.
Then $\tau(F^n\; \cdot\;)/\tau(F^n)\vert_{A_n}$ and
$\tau(e^{-nt_n(a)}\; \cdot\;)/\tau(e^{-nt_n(a)})\vert_{A_n}$
are states on $A_n$.
By the positivity of relative entropy, we have
\begin{align*}
0&\le
S\left(\frac{\tau(F^n\cdot)}{\tau(F^n)}\vert_{A_n},
\frac{\tau(e^{-nt_n(a)}\cdot)}{\tau(e^{-nt_n(a)})}\vert_{A_n}\right)\\
&=-\log\tau(F^n)+n\frac{\tau\lmk F^nt_n(a)\rmk}{\tau(F^n)}
+\log\tau\lmk e^{-n t_{n}(a)}\rmk\\
&\le
-\log\tau(F^n)+n\lmk \omega(a)+\delta\rmk
+\log\tau\lmk e^{-n t_{n}(a)}\rmk.
\end{align*}
This immediately implies (\ref{rk}).

Next for each $n$, we construct a projection $Q^n$
commuting with the density matrix $D_{\omega\vert_{A_n}}$
and satisfying $\tau(Q^n)=\tau(F^n)$ and $\omega(F^n)\le\omega(Q^n)$.
Write $A_n$ in the form
\[
A_n=\bigoplus_{k=1}^{N_n}M^n_k,
\]
where $M^n_k$ is isomorphic to a matrix algebra 
$\Mat_{m_{k}^n}\lmk{\mathbb C}\rmk$.
Let $z_k^n$ be the central projection in $A_n$ such that 
$M^n_k=A_nz^n_k$.
Since $\tau$ is trace, the density matrix $D_{\tau\vert_{A_n}}$ 
of $\tau\vert_{A_n}$ has the from
\[
D_{\tau\vert_{A_n}}=\bigoplus_{k=1}^{N_n}\lambda_k^n z_k^n
\]
with $0\le \lambda_k^n\le 1$.
Write $D_{\omega\vert_{A_n}}$ and $F^n$ as
\[
D_{\omega\vert_{A_n}}=\bigoplus_{k=1}^{N_n} D_k^n,\quad
F^n=\bigoplus_{k=1}^{N_n}F_k^n
\]
with positive elements $D_k^n\in {M^n_k}$ and projections
$F_k^n\in {M^n_k}$.

For each $k$, write $D_k^n$ in the form
$
D_k^n=\sum_{i=1}^{m_{k}^n}\beta_{i,k}^nq_{i,k}^n
$
where $\beta_{i,k}^n\ge 0$ and $q_{i,k}^n$ are mutually
orthogonal minimal projections in $A_n$ with 
$\sum_{i=1}^{m_{k}^n}q_{i,k}^n=z_{k}^n$.
We may assume $\beta_{1,k}^n\ge \beta_{i,k}^n\ge\cdots\ge \beta_{m_k^n,k}^n$.
Set $Q_k^n:=\sum_{i=1}^{\Tr_{M_k^n}F_k^n}q_{i,k}^n$.
Then, by Ky Fan's Theorem, 
\[\Tr_{M_k^n}(D_k^nQ_k^n)\ge\Tr_{M_k^n}(D_k^nF_k^n),\quad
\Tr_{M_k^n}Q_k^n=\Tr_{M_k^n}F_k^n.
\] 
The projection $Q^n:=\bigoplus_{k=1}^{N_n}Q_k^n$ satisfies the
required condition.

To complete the proof, note that $Q^n$ and $\Proj\left
[\frac 1n \lmk\log D_{\omega\vert_{A_n}}-\log D_{\tau\vert_{A_n}}\rmk 
\le t\right]$ commute.
Therefore, we have
\begin{align*}
&\omega\lmk\Proj\left
[\frac 1n \lmk\log D_{\omega\vert_{A_n}}-\log D_{\tau\vert_{A_n}}\rmk 
\le t\right]
\rmk\\
&=\omega\lmk\Proj\left
[\frac 1n \lmk\log D_{\omega\vert_{A_n}}-\log D_{\tau\vert_{A_n}}\rmk 
\le t\right]Q^n
\rmk\\
&+
\omega\lmk\Proj\left
[\frac 1n \lmk\log D_{\omega\vert_{A_n}}-\log D_{\tau\vert_{A_n}}\rmk 
\le t\right]\lmk 1-Q^n\rmk
\rmk\\
&\le e^{nt}\tau(Q^n)+\omega\lmk1-Q^n\rmk.
\end{align*}
It suffices to show that each term on the right hand side converges to
$0$ as $n\to\infty$. 
From $\tau(Q^n)=\tau(F^n)$ and (\ref{rk}), we have the exponentially fast decay
of the first term:
\begin{align*}
&\limsup_{n\to\infty}\frac 1n \log e^{nt}\tau(Q^n)
=t+\limsup_{n\to\infty}\frac 1n \log \tau(Q^n)\\
&=t+\limsup_{n\to\infty}\frac 1n \log \tau(F^n)
\le
t+P_{\gamma}^{\tau}(a)+\omega(a)+\delta<0.
\end{align*}
Furthermore, $\omega(F^n)\le\omega(Q^n)$ and
(\ref{oe}) implies $\lim_{n\to\infty}\omega\lmk1-Q^n\rmk=0$.
\end{proof}
\begin{proofof}[Theorem \ref{main}]
From the condition (vi) in Assumption \ref{nine},
it suffices to show that for any $\delta>0$,
\begin{align}\label{mid}
\lim_{n\to\infty}\omega
\lmk
\Proj
\left[
\frac 1n  \lmk\log
D_{\omega\vert_{A_n}}-\log D_{\tau\vert_{A_n}}\rmk\in
\lmk-s(\omega)+\lambda_{\tau}-\delta,-s(\omega)
+\lambda_{\tau}+\delta\rmk
\right]
\rmk=1.
\end{align}
Theorem \ref{main} follows from this with
\begin{align*}
p_n:=&\Proj
\left[
\frac 1n  \lmk\log
D_{\omega\vert_{A_n}}-\log D_{\tau\vert_{A_n}}\rmk\in
\lmk-s(\omega)+\lambda_{\tau}-\delta/3,-s(\omega)
+\lambda_{\tau}+\delta/3\rmk
\right]\\
&\cdot\Proj
\left[
\frac 1n  \log D_{\tau\vert_{A_n}}\in
\lmk-\lambda_{\tau}-\delta/3,-\lambda_{\tau}+\delta/3\rmk
\right].
\end{align*}
As we already have 
Lemma \ref{lb}, in order to show (\ref{mid}), it suffices to show that
for any $\delta>0$,
\begin{align*}
\lim_{n\to\infty}\omega
\lmk
\Proj
\left[
\frac 1n \log \lmk
D_{\omega\vert_{A_n}}-\log D_{\tau\vert_{A_n}}\rmk\ge
-s(\omega)
+\lambda_{\tau}+\delta
\right]
\rmk=0.
\end{align*}
Fix $\varepsilon>0$ and choose 
$\delta'>0$ with $\delta'/(\delta'+\delta)<\varepsilon$.
Set 
\begin{align*}
Q_{n}^-&:=\Proj\lcm\frac 1n\lmk
\log D_{\omega\vert_{A_n}}-\log D_{\tau\vert_{A_n}}\rmk
\le -s(\omega)
+\lambda_{\tau}-\delta'\rcm,\\
Q_{n}^+&:=\Proj\lcm\frac 1n\lmk
\log D_{\omega\vert_{A_n}}-\log D_{\tau\vert_{A_n}}\rmk
\ge -s(\omega)
+\lambda_{\tau}+\delta\rcm.
\end{align*}
By the positivity of relative entropy 
$S\lmk\omega(Q_{n}^-\; \cdot\;)/\omega(Q_{n}^-)\vert_{A_n},
\tau(Q_{n}^-\; \cdot\;)/{\tau(Q_{n}^-)}\vert_{A_n}\rmk$
and the condition (vi) in Assumption (\ref{nine}) 
we have
\begin{align*}
\frac{1}{n}\omega(Q_{n}^-)\log \omega(Q_{n}^-)
\le \frac{1}{n}\omega\lmk
Q_{n}^-\lmk
\log D_{\omega\vert_{A_n}}-\log D_{\tau\vert_{A_n}}
\rmk
\rmk\le
C\omega(Q_{n}^-),
\end{align*}
for $n$ large enough, with some positive constant $C$.
As we have $\lim_{n\to\infty}\omega(Q_{n}^-)=0$ 
from Lemma \ref{lb}, these inequalities implies
\begin{align}\label{nc}
\lim_{n\to\infty}
\frac{1}{n}\omega\lmk
Q_{n}^-
\lmk
\log D_{\omega\vert_{A_n}}-\log D_{\tau\vert_{A_n}}
\rmk\rmk=0.
\end{align}
By the definitions of $Q_{n}^-$, $Q_{n}^+$, we have
\begin{align*}
&\omega\lmk
\frac 1n \lmk 
\log D_{\omega\vert_{A_n}}-\log D_{\tau\vert_{A_n}}
\rmk
\rmk\\
&\ge
\frac 1n\omega\lmk Q_{n}^-\lmk
\log D_{\omega\vert_{A_n}}-\log D_{\tau\vert_{A_n}}
\rmk
\rmk\\
&\;\;\;\;+\lmk
\omega(Q_{n}^-)-1
\rmk\lmk s(\omega)
-\lambda_{\tau}+\delta'\rmk
+(\delta+\delta')\omega(Q_{n}^+).
\end{align*}
Taking $n\to\infty$ limit, we obtain
\begin{align*}
\limsup_n\omega(Q_{n}^+)\le\frac{\delta'}{\delta'+\delta}<\varepsilon.
\end{align*}
\end{proofof}
\section{Variational principle}\label{mevp}
In this section, we give a proof of Proposition \ref{me} 
and \ref{vp}. Most of the arguments are analogous to
those in quantum spin systems.
By the condition (vi) of Assumption \ref{nine}, in order to prove
Proposition \ref{me}
it suffices to consider
\[
-\frac1n S(\omega\vert_{A_n},\tau\vert_{A_n})
=\frac1n S(\omega\vert_{A_n})+\frac1n\omega(\log D_{\tau\vert_{A_n}})
\]
instead of $\frac1n S(\omega\vert_{A_n})$.

Fix some $m\in\nan$. 
Set $I_{m,j}:=[(m+n_0)j,(m+n_0)(j+1)-n_0-1]$ for each $j\in {\zin}$.
From (i) of Assumption \ref{nine},
there exists a $*$-homomorphism 
$\pi_m:\bigotimes_{\zin}A_m\to
\bigvee_{j\in\zin}A_{I_{m,j}}$ with
$\pi_m\vert_{\bigotimes\limits_{j=k}^lA_m}
(\bigotimes\limits_{j=k}^lb_j)=
\prod\limits_{j=k}^l\gamma^{(m+n_0)j}(b_j)$,
$b_j\in A_m$ for $j=k,\ldots,l$, $k\le l$.
Furthermore, (ii) of Assumption \ref{nine} implies
$\tau\circ\pi_m=\bigotimes_{\zin}\tau\vert_{A_m}$.
From this $\pi_m$ is $*$-isomorphism.
\begin{proofof}[Proposition \ref{me}]
For each $n\ge m+n_0+1$, we have
\begin{align*}
&-S\lmk
\left.\omega\right
\vert_{A_n},\left.\tau\right\vert_{A_n}\rmk
\le
-S\lmk
\left.\omega\right
\vert_{\bigvee\limits_{j=0}^{[\frac{n}{m+n_0}]-1}
A_{I_{m,j}}},\left.\tau\right\vert_{\bigvee\limits_{j=0}^{[\frac{n}{m+n_0}]-1}
A_{I_{m,j}}}\rmk\\
&=-S\lmk
\left.\omega\circ\pi_m\right
\vert_{\bigotimes\limits_{j=0}^{[\frac{n}{m+n_0}]-1}A_{m}}
,
\left.\tau\circ\pi_m
\right
\vert_{\bigotimes\limits_{j=0}^{[\frac{n}{m+n_0}]-1}A_m}\rmk
\le-
\left[
\frac{n}{m+n_0}
\right]S\lmk
\left.\omega\right
\vert_{A_m},\left.\tau
\right\vert_{A_m}\rmk.
\end{align*}
We used monotonicity of relative entropy in the first inequality and
the subadditivity of von Neumann entropy in the 
second inequality.
This inequality and Lemma 1.1.2 of \cite{nsb} imply
the existence of the limit 
$\lim_{n\to\infty}-\frac1n S(\omega\vert_{A_n},\tau\vert_{A_n})
=\inf_m -S(\omega\vert_{A_m},\tau\vert_{A_m})/(m+n_0)$.
The rest of the proposition is standard. 
We leave it to the reader to check it.
(See \cite{BR96}, \cite{hp}.)
\end{proofof}
\begin{lem}
For any self-adjoint element
 $a$ in $A$, the limit
\[
P_{\gamma}^{\tau}(a):=\lim_{n\to\infty} \frac 1n \log 
\tau(e^{-\sum_{i=0}^{n-1}\gamma_i(a)})\] exists.
The function $P_{\gamma}^{\tau}$ is continuous with respect to the 
norm topology on $A_{sa}$.
\end{lem}
\begin{proof}
From Corollary 2.3.11 of \cite{nsb}, 
we may assume that $a\in A_l$ for some $l\in\nan$.
Fix some $m\in\nan$, with $l\le m$. Then from the condition (i), (ii)
of Assumption \ref{nine}, and again Corollary 2.3.11 of \cite{nsb},
we have
\begin{align*}
&\log\tau\lmk e^{-\sum_{i=0}^{n-1}\gamma_i(a)}\rmk\\
&\le
\left[\frac{n}{m+n_0}\right]\log
\tau\lmk e^{-\sum_{i=0}^{m-l}\gamma_i(a)}\rmk
+\left[\frac{n}{m+n_0}\right]\lmk
n_0+l-1\rmk\lV a\rV+\lmk m+n_0\rmk\lV a\rV,
\end{align*}
for all $m+n_0\le n\in\nan$.
This inequality and Lemma 1.1.2 of \cite{nsb} implies
the existence of 
$P_{\gamma}^{\tau}(a)$. Continuity of $P_{\gamma}^{\tau}$
follows from Corollary 2.3.11 of \cite{nsb}.
(See \cite{hp}.)
\end{proof}
\begin{proofof}[Proposition \ref{vp}]
The inequality
\[
P_{\gamma}^{\tau}(a)
\ge\sup_{\omega\in {\cal S}_{\gamma}}\left\{
s(\omega)-\omega(a)\right\}-\lambda_{\tau}
\]
follows from positivity of relative entropy
$S(\omega\vert_{A_n}, 
\tau(e^{-\sum_{i=0}^{n-l}\gamma_i(a)}\;\cdot\;)/
\tau(e^{-\sum_{i=0}^{n-l}\gamma_i(a)})\vert_{A_n})$.

Since $P_{\gamma}^{\tau}$ is continuous, to prove the converse
inequality we may assume $a\in A_l$ for some $l\in\nan$.
Fix $l\le m\in\nan$.
From the condition (i), (ii) of Assumption \ref{nine},
there exists a $\gamma_{m+n_0}$-invariant state 
$\varphi_{m,a}$
given by
\[
\varphi_{m,a}=\wks\lim_{N\to\infty}
\frac{\tau\lmk e^{-\sum_{j=-N}^N\sum_{i=0}^{m-l}
\gamma_{j(m+n_0)+i}(a)}
\;\cdot\;\rmk}
{\tau\lmk e^{-\sum_{j=-N}^N\sum_{i=0}^{m-l}
\gamma_{j(m+n_0)+i}(a)}
\rmk}.
\]

The restriction of $\varphi_{m,a}$ to 
$A_{[k(m+n_0), n(m+n_0)+m]}$ with $k\le n$, $k,n\in\zin$
is
\[
\left.\varphi_{m,a}\right\vert_{A_{[k(m+n_0), n(m+n_0)+m]}}
=\left.
\frac{\tau\lmk e^{-\sum_{j=k}^n\sum_{i=0}^{m-l}
\gamma_{j(m+n_0)+i}(a)}
\;\cdot\;\rmk}
{\tau\lmk e^{-\sum_{j=k}^n\sum_{i=0}^{m-l}
\gamma_{j(m+n_0)+i}(a)}
\rmk}\right\vert_{A_{[k(m+n_0), n(m+n_0)+m]}}.
\]
By monotonicity and positivity of relative entropy,
we have
\begin{align*}
0
&\le
S\lmk
\left.\varphi_{m,a}\circ
\gamma_i\right\vert_{A[0,n(m+n_0)-1]}, 
\left.\frac{\tau\lmk
e^{-\sum_{k=0}^{n(m+n_0)-l-1}\gamma_k(a)}\;\cdot\;\rmk}
{\tau\lmk
e^{-\sum_{k=0}^{n(m+n_0)-l-1}\gamma_k(a)}\rmk}
\right\vert_{A_{[0,n(m+n_0)-1]}}
\rmk\\
&\le
S\lmk
\left.\varphi_{m,a}\circ
\gamma_i\right\vert_{A[-i,(n+1)(m+n_0)+m-i]}, 
\left.\frac{\tau\lmk
e^{-\sum_{k=0}^{n(m+n_0)-l-1}\gamma_k(a)}\;\cdot\;\rmk}
{\tau\lmk
e^{-\sum_{k=0}^{n(m+n_0)-l-1}\gamma_k(a)}\rmk}
\right\vert_{A_{[-i,(n+1)(m+n_0)+m-i]}}
\rmk,
\end{align*}
for each $i=0,\ldots,m+n_0-1$ and $n\in\nan$.
From this we get
\begin{align}\label{sw}
0&\le
-S(\left. \varphi_{m,a}\circ\gamma_i\right
\vert_{A_{n(m+n_0)}})
-\varphi_{m,a}\circ\gamma_i
\lmk \log D_{\tau\vert_{A_{n(m+n_0)}}}\rmk\nonumber\\
&+\log\tau\lmk e^{-\sum_{k=0}^{n(m+n_0)-l-1}\gamma_k(a)}
\rmk
+\varphi_{m,a}\lmk
\sum_{k=i}^{n(m+n_0)-l+i-1}\gamma_k(a)\rmk\nonumber\\
&\le
2\lV a\rV \lmk
4(m+n_0)+n(n_0+l)\rmk,
\end{align}
for each $i=0,\ldots,m+n_0-1$ and $n\in\nan$.

Note that for a $\gamma_{m+n_0}$-invariant state $\varphi$,
the mean entropy with respect to  $\gamma_{m+n_0}$ exists:
\[
s^{m+n_0}(\varphi):=\lim_{n\to\infty} 
\frac{1}{n}S\lmk\varphi\vert_{A_{n(m+n_0)}}\rmk.
\]
This $s^{m+n_0}$ is also 
affine on the set of all $\gamma_{m+n_0}$-invariant states. 
If $\varphi$ is $\gamma$-invariant, we have 
$s^{m+n_0}(\varphi)=(m+n_0)s(\varphi)$.

From (\ref{sw}), we have
\begin{align}\label{ii}
s^{m+n_0}\lmk\varphi_{m,a}\circ\gamma_i\rmk-(m+n_0)\lambda_{\tau}
\ge
(m+n_0)P_{\gamma}^{\tau}(a)+\sum_{k=0}^{m+n_0-1}
\varphi_{m,a}\circ\gamma_k(a)-2\lV a\rV(n_0+l).
\end{align}

Define a $\gamma$-invariant state $\psi_{m,a}$ by
$\psi_{m,a}:=\lmk 
\sum_{i=0}^{m+n_0-1}\varphi_{m,a}\circ\gamma_i\rmk/(m+n_0)$.
As $s^{m+n_0}$ is affine and $\psi_{m,a}$ is $\gamma$-invariant,
(\ref{ii}) implies
\begin{align*}
s\lmk\psi_{m,a}\rmk&=\frac{1}{m+n_0}s^{m+n_0}\lmk\psi_{m,a}\rmk
=\frac{1}{(m+n_0)^2}\sum_{i=0}^{m+n_0-1}
s^{m+n_0}\lmk\varphi_{m,a}\circ\gamma_i\rmk\\
&\ge
P_{\gamma}^{\tau}(a)+\lambda_{\tau}
+\psi_{m,a}(a)-\frac{2\lV a\rV(n_0+l)}{m+n_0}.
\end{align*}
Hence we obtain the converse inequality
\[
P_{\gamma}^{\tau}(a)
\le\sup_{\omega\in {\cal S}_{\gamma}}\left\{
s(\omega)-\omega(a)\right\}-\lambda_{\tau}.
\]

\end{proofof}

\end{document}